\documentclass[11pt, twoside, a4paper]{article}

\usepackage{amsmath}     				
\usepackage{amssymb}   					
\usepackage{amsthm}     				
\usepackage{graphicx}    				

\usepackage{textcomp}    				
\usepackage[T1]{fontenc} 				
\usepackage{marvosym}    				
\usepackage[sc]{mathpazo}  	 			

\usepackage{authblk} 					
\usepackage[usenames]{xcolor}  			
\usepackage{lastpage} 					

\usepackage{enumerate}	 				
\usepackage{mathtools}					
\usepackage{bm}						
\usepackage[noadjust]{cite}			

\usepackage{hyperref} 					
\usepackage[all]{hypcap}				

\definecolor{ForestGreen}{rgb}{0.15,0.416,0.18}
\definecolor{EgyptBlue}{rgb}{0.063,0.2,0.65}
\hypersetup{
	colorlinks=true,
	linkcolor=EgyptBlue,
   citecolor=EgyptBlue,
   urlcolor=ForestGreen
}

\newtheorem{theorem}{Theorem}[section]

\theoremstyle{definition}

\theoremstyle{definition}
\newtheorem{remark}[theorem]{Remark}
\theoremstyle{definition}

\numberwithin{equation}{section}
\numberwithin{table}{section}
\numberwithin{figure}{section}

\title{Solvable difference equations similar to the Newton--Raphson iteration for algebraic equations}

\author{\textbf{Kazuki Maeda}\footnote{Email: kmaeda@kmaeda.net}}
\affil{Faculty of Informatics, University of Fukuchiyama,
  3370 Hori, Fukuchiyama,
  Kyoto 620-0886,
  Japan}

\newcommand\shorttitle{Solvable difference equations similar to the Newton--Raphson iteration for algebraic equations}

\newcommand\authorsshort{K. Maeda}

\makeatletter
\renewcommand{\maketitle}{\bgroup\setlength{\parindent}{0pt}

\vspace{1truecm}
\begin{center}{\vbox{\titlefont\@title}}\end{center}
\vspace{0.5truecm}
\begin{center}{\@author} \end{center}

\egroup
}

\renewcommand{\@fnsymbol}[1]{%
    \ifcase#1 \or {\,\Letter\!} \or\textasteriskcentered\or \textasteriskcentered\textasteriskcentered
    \else\@ctrerr\fi}
\makeatother

\makeatletter
\newcommand{\hbibitem}[4]{\bibitem{#1}{#2}
\def\@tempa{#3}%
\def\@tempb{#4}%
\ifx\@tempa\@empty\ifx\@tempb\@empty{}{}\else{}{\href{http://dx.doi.org/#4}{url}}\fi\else
{\href{http://www.ams.org/mathscinet-getitem?mr=#3}{MR#3}}\ifx\@tempb\@empty{}\else{; \href{http://dx.doi.org/#4}{url}}\fi\fi}
\makeatother

\newcommand*{\titlefont}{\fontsize{18}{21.6}\selectfont\textbf}

\makeatletter
\renewcommand\@author{\ifx\AB@affillist\AB@empty\AB@author\else
      \ifnum\value{affil}>\value{Maxaffil}\def\rlap##1{##1}%
    \AB@authlist\\[\affilsep]\vbox{\AB@affillist}
    \else  \AB@authors\fi\fi}
\makeatother

\makeatletter

\def\ps@plain{
\def\@oddhead{\ifnum\thepage=1
\else
\hss\textit{\shorttitle}\hss\hbox to 0pt{\hss\thepage}\fi}\def\@oddfoot{}
\def\@evenhead{\hbox to 0pt{\thepage\hss} \hss\textit{\authorsshort}\hss}
\def\@evenfoot{}
}
\makeatother

\begin{document}

\maketitle

\pagestyle{plain}

\begin{center}
\noindent
\begin{minipage}{0.85\textwidth}


\bigskip

{\small{
\noindent {\bf Abstract.}
It is known that difference equations
generated as the Newton--Raphson iteration
for quadratic equations are solvable in closed form,
and the solution can be constructed from
linear three-term recurrence relations with constant coefficients.
We show that the same construction for four-term recurrence
relations gives the solution to the initial value problem
of difference equations similar to the Newton--Raphson iteration
for cubic equations.
In many cases,
the solution converges to a root of the cubic equation
and the convergence rate is quadratic.
  }
\smallskip

\smallskip

}

\end{minipage}
\end{center}

\section{Introduction} 

For a given function $f\colon \mathbf C\to\mathbf C$,
let us consider the equation $f(z)=0$.
To solve this equation numerically,
we often use the Newton--Raphson iteration method
\begin{equation*}
  z_{n+1}=z_n-\frac{f(z_n)}{f'(z_n)},\quad
  n=0, 1, 2, \dots.
\end{equation*}
In many cases, $z_n$ converges to a solution of the equation $f(z)=0$ as $n\to\infty$, where the convergence rate is quadratic.

For example, in the case of $f(z)=z^2+c_1z+c_2$, where $c_1, c_2 \in \mathbf C$ are some constants,
the Newton--Raphson iteration is
\begin{equation}\label{eq:newton-quad}
  z_{n+1}=z_n-\frac{z_n^2+c_1z_n+c_2}{2z_n+c_1},\quad
  n=0, 1, 2, \dots.
\end{equation}
It is known that the difference equation \eqref{eq:newton-quad} is
solvable in closed form (e.g., Stevi\'c~\cite{stevic2023scs} reviewed the solution and discussed several related topics).
Suppose that $c_1^2-4c_2\ne 0$, i.e., the quadratic equation
$z^2+c_1z+c_2=0$ has two distinct roots $\lambda_0$ and $\lambda_1$.
Then, the solution to the initial value problem of the difference
equation \eqref{eq:newton-quad} is given by
\begin{equation}\label{eq:sol-newton-quad}
  z_n=\frac{\lambda_1(z_0-\lambda_0)^{2^n}-\lambda_0(z_0-\lambda_1)^{2^n}}{(z_0-\lambda_0)^{2^n}-(z_0-\lambda_1)^{2^n}}.
\end{equation}
This solution explicitly shows that $z_n$ goes to a root of
the quadratic equation as $n\to \infty$ and
the convergence rate is quadratic.

In the case of $f(z)=z^3+c_1z^2+c_2z+c_3$, where $c_1, c_2, c_3\in\mathbf C$
are some constants, the Newton--Raphson iteration is
\begin{equation}\label{eq:newton-cubic}
  z_{n+1}=z_n-\frac{z_n^3+c_1z_n^2+c_2z_n+c_3}{3z_n^2+2c_1z_n+c_2},\quad
  n=0, 1, 2, \dots.
\end{equation}
In contrast to the quadratic case,
the difference equation \eqref{eq:newton-cubic} does not seem
to have closed form solutions.

A proof of the fact that \eqref{eq:sol-newton-quad} is the solution of the difference equation~\eqref{eq:newton-quad}
is direct substitution.
Another proof was given by Kondo and Nakamura~\cite{kondo2002dss} in which they constructed
the difference equation~\eqref{eq:newton-quad} from an addition formula for determinants of tridiagonal matrices,
which is a solution of a three-term linear recurrence relation.
In this paper, we will try to extend this construction to a four-term linear recurrence relation
and obtain a solvable iterative method for cubic equations which may be similar to the Newton--Raphson iteration.

This paper is organized as follows.
In section~\ref{sec:deriv-solut-newt},
we will review a construction of the solution to
the initial value problem of the difference equation
\eqref{eq:newton-quad}
using a three-term linear recurrence relation
presented by Kondo--Nakamura.
The method presented in this paper is simplified and systematized than the original.
Thanks to that, we can easily extend the method to another case.
In section~\ref{sec:constr-solv-diff},
we will extend the construction to a four-term linear
recurrence relation and obtain a system of difference equations
which is similar to the one generated by the Newton--Raphson
iteration method for cubic equations.
Section~\ref{sec:concluding-remarks} is concluding remarks.

\section{Derivation of the solution to the Newton--Raphson iteration for quadratic equations}\label{sec:deriv-solut-newt}

In this section, we briefly explain the construction of the solution
originally presented by Kondo--Nakamura~\cite{kondo2002dss}.

\subsection{Three-term recurrence relations and the Newton--Raphson iteration}

Let us consider a sequence $(\phi_n)_{n=0}^\infty$
defined by the following three-term recurrence relation
\begin{equation}\label{eq:three-term}
  \phi_{-1}\coloneq 0,\quad
  \phi_{0}\coloneq 1,\quad
  \phi_{n+1}\coloneq-\gamma_1\phi_n-\gamma_2\phi_{n-1},\quad
  n=0, 1, 2, \dots,
\end{equation}
where $\gamma_1, \gamma_2 \in \mathbf C$ are some constants.

\begin{theorem}[Kondo--Nakamura~\cite{kondo2002dss}]\label{thm:three-term}
  The sequence $(\phi_n)_{n=0}^\infty$
  defined by the three-term recurrence relation~\eqref{eq:three-term}
  satisfies the relation
  \begin{equation}\label{eq:three-term-2}
    \phi_{m+n}=\phi_m\phi_n-\gamma_2\phi_{m-1}\phi_{n-1}
  \end{equation}
  for any non-negative $m, n \in \mathbf Z$.
\end{theorem}
\begin{proof}
  Since $\phi_{-1}=0$, $\phi_0=1$, and $\phi_1=-\gamma_1\phi_0-\gamma_2\phi_{-1}=-\gamma_1$,
  the three-term recurrence relation~\eqref{eq:three-term}
  is rewritten using matrices and vectors as
  \begin{equation*}
    \begin{pmatrix}
      \phi_{n+1}\\
      \phi_n
    \end{pmatrix}
    =
    \begin{pmatrix}
      -\gamma_1 & -\gamma_2\\
      1 & 0
    \end{pmatrix}
    \begin{pmatrix}
      \phi_n\\
      \phi_{n-1}
    \end{pmatrix}
    =
    \begin{pmatrix}
      \phi_1 & -\gamma_2\phi_0\\
      \phi_0 & -\gamma_2\phi_{-1}
    \end{pmatrix}
    \begin{pmatrix}
      \phi_n\\
      \phi_{n-1}
    \end{pmatrix}.
  \end{equation*}
  Hence, for any non-negative $m, n \in \mathbf Z$,
  we obtain
  \begin{align*}
    \begin{pmatrix}
      \phi_{m+n}\\
      \phi_{m+n-1}
    \end{pmatrix}
    &=
    \begin{pmatrix}
      -\gamma_1 & -\gamma_2\\
      1 & 0
    \end{pmatrix}^m
    \begin{pmatrix}
      \phi_n\\
      \phi_{n-1}
    \end{pmatrix}\\
    &=
    \begin{pmatrix}
      -\gamma_1 & -\gamma_2\\
      1 & 0
    \end{pmatrix}^{m-1}
    \begin{pmatrix}
      \phi_1 & -\gamma_2\phi_0\\
      \phi_0 & -\gamma_2\phi_{-1}
    \end{pmatrix}
    \begin{pmatrix}
      \phi_n\\
      \phi_{n-1}
    \end{pmatrix}\\
    &=
    \begin{pmatrix}
      \phi_{m} & -\gamma_2\phi_{m-1}\\
      \phi_{m-1} & -\gamma_2\phi_{m-2}
    \end{pmatrix}
    \begin{pmatrix}
      \phi_n\\
      \phi_{n-1}
    \end{pmatrix}\\
    &=
    \begin{pmatrix}
      \phi_m\phi_n-\gamma_2\phi_{m-1}\phi_{n-1}\\
      \phi_{m-1}\phi_n-\gamma_2\phi_{m-2}\phi_{n-1}
    \end{pmatrix}.
  \end{align*}
  This completes the proof.
\end{proof}

Let us introduce a variable $\xi_n\coloneq \frac{\phi_n}{\phi_{n-1}}$.
Then, since $\phi_n=\xi_n\phi_{n-1}$,
the three-term recurrence relation \eqref{eq:three-term}
gives $-\gamma_2\phi_{n-2}=\phi_n+\gamma_1\phi_{n-1}=(\xi_n+\gamma_1)\phi_{n-1}$.
By using the relation \eqref{eq:three-term-2} and these relations
\begin{equation*}
  \phi_n=\xi_n\phi_{n-1},\quad
  -\gamma_2\phi_{n-2}=(\xi_n+\gamma_1)\phi_{n-1},
\end{equation*}
we obtain
\begin{equation}\label{eq:xi-phi}
  \hspace{-0.5em}
  \xi_{m+n}
  =\frac{\phi_{m+n}}{\phi_{m+n-1}}
  =\frac{\phi_m\phi_n-\gamma_2\phi_{m-1}\phi_{n-1}}{\phi_{m-1}\phi_n-\gamma_2\phi_{m-2}\phi_{n-1}}
  =\frac{(\xi_m\xi_n-\gamma_2)\phi_{m-1}\phi_{n-1}}{\left(\xi_n+(\xi_m+\gamma_1)\right)\phi_{m-1}\phi_{n-1}}
  =\frac{\xi_m\xi_n-\gamma_2}{\xi_m+\xi_n+\gamma_1}.
\end{equation}
Especially, if $m=n$, then the relation yields
\begin{equation}\label{eq:xi-rec}
  \xi_{2n}=\frac{\xi_n^2-\gamma_2}{2\xi_n+\gamma_1},
\end{equation}
and let $x_n=\xi_{2^n}$, then we obtain
\begin{equation*}
  x_{n+1}=\frac{x_n^2-\gamma_2}{2x_n+\gamma_1}
  =x_n-\frac{x_n^2+\gamma_1x_n+\gamma_2}{2x_n+\gamma_1}.
\end{equation*}
This equation is the same form as the Newton--Raphson iteration
for quadratic equations~\eqref{eq:newton-quad}.
However,
since the initial value
$x_0=\xi_1=\frac{\phi_1}{\phi_0}=-\gamma_1$ is fixed,
$x_n=\xi_{2^n}=\frac{\phi_{2^n}}{\phi_{2^n-1}}$ cannot give
the solution for \emph{all} initial values.

\subsection{The solution to the initial value problem}
To obtain the solution for any initial value,
consider a variable transformation
$\xi_n=\eta_n-\eta_1+\xi_1=\eta_n-\eta_1-\gamma_1$,
where one can choose any constant $\eta_1 \in \mathbf C$.
Then, the equation~\eqref{eq:xi-rec} is transformed to
\begin{equation*}
  \eta_{2n}=\eta_n-\frac{\eta_n^2+(-2\eta_1-\gamma_1)\eta_n+\eta_1^2+\gamma_1\eta_1+\gamma_2}{2\eta_n-2\eta_1-\gamma_1}.
\end{equation*}
Hence, setting
\begin{equation}\label{eq:quad-c}
  c_1\coloneq -2\eta_1-\gamma_1,\quad
  c_2\coloneq \eta_1^2+\gamma_1\eta_1+\gamma_2,\quad
  z_n\coloneq \eta_{2^n}
\end{equation}
yields
\begin{equation*}
  z_{n+1}=z_n-\frac{z_n^2+c_1z_n+c_2}{2z_n+c_1}.
\end{equation*}
Since one can choose any initial value $z_0=\eta_1$,
the solution to the initial value problem is given by
\begin{equation}\label{eq:quad-sol}
  z_n=\eta_{2^n}=\xi_{2^n}+\eta_1+\gamma_1=\frac{\phi_{2^n}}{\phi_{2^n-1}}+z_0+\gamma_1.
\end{equation}

Let us solve the three-term recurrence relation~\eqref{eq:three-term},
using the theory of linear difference equations with constant coefficients
(see a textbook of difference equations, e.g. Elaydi~\cite{elaydi2005ide}),
in terms of roots of the quadratic equation $z^2+c_1z+c_2=0$.
First, suppose that $c_1^2-4c_2\ne 0$ and
let $\lambda_0$ and $\lambda_1$ be
two distinct roots of the quadratic equation
$z^2+c_1z+c_2=0$.
In this case, since $c_1=-\lambda_0-\lambda_1$ and
$c_2=\lambda_0\lambda_1$, and from \eqref{eq:quad-c},
\begin{equation*}
  \gamma_1=-2z_0-c_1=-2z_0+\lambda_0+\lambda_1,\quad
  \gamma_2=z_0^2+c_1z_0+c_2=z_0^2+(-\lambda_0-\lambda_1)z_0+\lambda_0\lambda_1.
\end{equation*}
Therefore, the characteristic equation of
the three-term recurrence relation~\eqref{eq:three-term} is factored as
\begin{align*}
  z^2+\gamma_1z+\gamma_2
  &=z^2+(-2z_0+\lambda_0+\lambda_1)z+z_0^2+(-\lambda_0-\lambda_1)z_0+\lambda_0\lambda_1\\
  &=(z-(z_0-\lambda_0))(z-(z_0-\lambda_1))=0.
\end{align*}
This implies that $\phi_n$ is of the form
\begin{equation*}
  \phi_n=C_0(z_0-\lambda_0)^n+C_1(z_0-\lambda_1)^n,
\end{equation*}
where the coefficients $C_0$ and $C_1$ are determined by
\begin{equation*}
  \phi_0=C_0+C_1=1,\quad
  \phi_1=C_0(z_0-\lambda_0)+C_1(z_0-\lambda_1)=-\gamma_1=2z_0-\lambda_0-\lambda_1,
\end{equation*}
i.e.
\begin{equation*}
  C_0=\frac{z_0-\lambda_0}{\lambda_1-\lambda_0},\quad
  C_1=-\frac{z_0-\lambda_1}{\lambda_1-\lambda_0}.
\end{equation*}
Now we have
\begin{equation*}
  \phi_n=\frac{(z_0-\lambda_0)^{n+1}-(z_0-\lambda_1)^{n+1}}{\lambda_1-\lambda_0}
\end{equation*}
and substituting it into \eqref{eq:quad-sol} yields
\begin{equation*}
  z_n
  =\frac{(z_0-\lambda_0)^{2^n+1}-(z_0-\lambda_1)^{2^n+1}}{(z_0-\lambda_0)^{2^n}-(z_0-\lambda_1)^{2^n}}+z_0-2z_0+\lambda_0+\lambda_1
  =\frac{\lambda_1(z_0-\lambda_0)^{2^n}-\lambda_0(z_0-\lambda_1)^{2^n}}{(z_0-\lambda_0)^{2^n}-(z_0-\lambda_1)^{2^n}}.
\end{equation*}
If $|z_0-\lambda_0|<|z_0-\lambda_1|$,
then
\begin{equation*}
  z_n=\frac{\lambda_1\left(\frac{z_0-\lambda_0}{z_0-\lambda_1}\right)^{2^n}-\lambda_0}{\left(\frac{z_0-\lambda_0}{z_0-\lambda_1}\right)^{2^n}-1}\to \lambda_0\quad
  \text{as $n\to\infty$},
\end{equation*}
i.e. $z_n$ converges to the nearest root from the initial value $z_0$
and the convergence rate is quadratic.
If $|z_0-\lambda_0|=|z_0-\lambda_1|$,
then $z_n$ does not converge.

Next, suppose that $c_1^2-4c_2=0$ and let $\lambda$ be
a unique root of the quadratic equation $z^2+c_1z+c_2=0$.
In this case, the characteristic equation of
the three-term recurrence relation~\eqref{eq:three-term} is
\begin{equation*}
  z^2+\gamma_1 z+\gamma_2
  =(z-(z_0-\lambda))^2=0.
\end{equation*}
This implies that $\phi_n$ is of the form
\begin{equation*}
  \phi_n=C_0(z_0-\lambda)^n+C_1 n(z_0-\lambda)^{n-1},
\end{equation*}
where the coefficients $C_0$ and $C_1$ are determined by
\begin{equation*}
  \phi_0=C_0=1,\quad
  \phi_1=C_0(z_0-\lambda)+C_1=2z_0-2\lambda,
\end{equation*}
i.e.
\begin{equation*}
  C_0=1,\quad C_1=z_0-\lambda.
\end{equation*}
Now we have
\begin{equation*}
  \phi_n=(1+n)(z_0-\lambda)^n,
\end{equation*}
and substituting it into \eqref{eq:quad-sol} yields
\begin{equation*}
  z_n
  =\frac{(1+2^n)(z_0-\lambda)^{2^n}}{(1+2^n-1)(z_0-\lambda)^{2^n-1}}+z_0-2z_0+2\lambda
  =\frac{(2^n-1)\lambda+z_0}{2^n}.
\end{equation*}
Hence, $z_n\to \lambda$ as $n\to\infty$
and the convergence rate is linear.

\begin{remark}
  Kondo--Nakamura~\cite{kondo2002dss} pointed out that
  we can construct solutions to
  the iterative root finding methods which have higher order convergence rates,
  known as Halley's method, Kiss's method, Kiss--Nourein's method, ..., for quadratic equations
  in the same manner.

  For example, the Halley's iteration for the equation $f(z)=0$ is given by
  \begin{equation*}
    z_{n+1}=z_n-\frac{2f(z_n)f'(z_n)}{2(f'(z_n))^2-f(z_n)f''(z_n)}.
  \end{equation*}
  The iteration for the quadratic case $f(z)=z^2+c_1z+c_2$ is
  \begin{equation}\label{eq:halley-quad}
    z_{n+1}=z_n-\frac{2z_n^3+3c_1z_n^2+(c_1^2+2c_2)z_n+c_1c_2}{3z_n^2+3c_1z_n+c_1^2-c_2}.
  \end{equation}
  On the other hand, setting $m=2n$ in the equation~\eqref{eq:xi-phi} yields
  \begin{equation*}
    \xi_{3n}=\frac{\xi_{2n}\xi_n-\gamma_2}{\xi_{2n}+\xi_n+\gamma_1}
    =\xi_n-\frac{\tilde f(\xi_n)}{\xi_{2n}+\tilde f'(\xi_n)-\xi_n},
  \end{equation*}
  where $\tilde f(\xi_n)=\xi_n^2+\gamma_1\xi_n+\gamma_2$.
  Since~\eqref{eq:xi-rec} is rewritten as
  \begin{equation*}
    \xi_{2n}=\frac{\xi_n\tilde f'(\xi_n)-\tilde f(\xi_n)}{\tilde f'(\xi_n)}
  \end{equation*}
  and $\tilde f''(\xi_n)=2$, we obtain
  \begin{equation*}
    \xi_{3n}
    =\xi_n-\frac{\tilde f(\xi_n)\tilde f'(\xi_n)}{\xi_n\tilde f'(\xi_n)-\tilde f(\xi_n)+(\tilde f'(\xi_n))^2-\xi_n\tilde f'(\xi_n)}
    =\xi_n-\frac{2\tilde f(\xi_n)\tilde f'(\xi_n)}{2(\tilde f'(\xi_n))^2-\tilde f(\xi_n)\tilde f''(\xi_n)}.
  \end{equation*}
  Hence, we can readily obtain the solution to the initial value problem of the difference equation~\eqref{eq:halley-quad}:
  if $c_1^2-4c_2\ne 0$, then
  \begin{equation*}
    z_n=\frac{\lambda_1(z_0-\lambda_0)^{3^n}-\lambda_0(z_0-\lambda_1)^{3^n}}{(z_0-\lambda_0)^{3^n}-(z_0-\lambda_1)^{3^n}},
  \end{equation*}
  and if $c_1^2-4c_2=0$, then
  \begin{equation*}
    z_n=\frac{(3^n-1)\lambda+z_0}{3^n}.
  \end{equation*}
\end{remark}

\section{Construction of solvable difference equations similar to the Newton--Raphson iteration for cubic equations}\label{sec:constr-solv-diff}

In this section, we extend the construction explained in the previous section to the case of cubic equations.

\subsection{Four-term recurrence relations and rational difference equations}

Let us consider a sequence $(\phi_n)_{n=0}^\infty$ defined
by the following four-term recurrence relation
\begin{equation}\label{eq:four-term}
  \phi_{-2}=\phi_{-1}\coloneq 0,\quad
  \phi_0\coloneq 1,\quad
  \phi_{n+1}\coloneq-\gamma_1\phi_n-\gamma_2\phi_{n-1}-\gamma_3\phi_{n-2},\quad
  n=0, 1, 2, \dots,
\end{equation}
where $\gamma_1, \gamma_2, \gamma_3\in\mathbf C$ are some constants.
In addition, to prove the following theorem,
we set $\phi_{-3}\coloneq-\frac{1}{\gamma_3}$,
which is consistent with the four-term recurrence relation.

\begin{theorem}
  The sequence $(\phi_n)_{n=0}^\infty$ defined by
  the four-term recurrence relation~\eqref{eq:four-term}
  satisfies the relation
  \begin{equation}\label{eq:four-term-2}
    \phi_{m+n}=\phi_m\phi_n-\gamma_2\phi_{m-1}\phi_{n-1}-\gamma_3\phi_{m-2}\phi_{n-1}-\gamma_3\phi_{m-1}\phi_{n-2}
  \end{equation}
  for any non-negative $m, n\in\mathbf Z$.
\end{theorem}
\begin{proof}
  In the same manner as the proof of Theorem~\ref{thm:three-term},
  the four-term recurrence relation~\eqref{eq:four-term}
  is rewritten as
  \begin{equation*}
    \begin{pmatrix}
      \phi_{n+1}\\
      \phi_{n}\\
      \phi_{n-1}
    \end{pmatrix}=
    \begin{pmatrix}
      -\gamma_1 & -\gamma_2 & -\gamma_3\\
      1 & 0 & 0\\
      0 & 1 & 0
    \end{pmatrix}
    \begin{pmatrix}
      \phi_{n}\\
      \phi_{n-1}\\
      \phi_{n-2}
    \end{pmatrix}
    =
    \begin{pmatrix}
      \phi_1 & -\gamma_2\phi_0-\gamma_3\phi_{-1} & -\gamma_3\phi_0\\
      \phi_0 & -\gamma_2\phi_{-1}-\gamma_3\phi_{-2} & -\gamma_3\phi_{-1}\\
      \phi_{-1} & -\gamma_2\phi_{-2}-\gamma_3\phi_{-3} & -\gamma_3\phi_{-2}
    \end{pmatrix}
    \begin{pmatrix}
      \phi_{n}\\
      \phi_{n-1}\\
      \phi_{n-2}
    \end{pmatrix}.
  \end{equation*}
  Hence, we obtain
  \begin{equation*}
    \begin{pmatrix}
      \phi_{m+n}\\
      \phi_{m+n-1}\\
      \phi_{m+n-2}
    \end{pmatrix}
    =
    \begin{pmatrix}
      \phi_m & -\gamma_2\phi_{m-1}-\gamma_3\phi_{m-2} & -\gamma_3\phi_{m-1}\\
      \phi_{m-1} & -\gamma_2\phi_{m-2}-\gamma_3\phi_{m-3} & -\gamma_3\phi_{m-2}\\
      \phi_{m-2} & -\gamma_2\phi_{m-3}-\gamma_3\phi_{m-4} & -\gamma_3\phi_{m-3}
    \end{pmatrix}
    \begin{pmatrix}
      \phi_n\\
      \phi_{n-1}\\
      \phi_{n-2}
    \end{pmatrix}.
  \end{equation*}
  This completes the proof.
\end{proof}

Let us introduce a variable $\xi_n\coloneq\frac{\phi_n}{\phi_{n-1}}$.
Then, the four-term recurrence relation~\eqref{eq:four-term} gives
$-\gamma_3\phi_{n-2}=\phi_{n+1}+\gamma_1\phi_n+\gamma_2\phi_{n-1}=(\xi_{n+1}\xi_n+\gamma_1\xi_n+\gamma_2)\phi_{n-1}$.
By using the relation~\eqref{eq:four-term-2} and
these relations
\begin{equation*}
  \phi_n=\xi_n\phi_{n-1},\quad
  -\gamma_3\phi_{n-2}=(\xi_{n+1}\xi_n+\gamma_1\xi_n+\gamma_2)\phi_{n-1},
\end{equation*}
we obtain
\begin{align*}
  \xi_{m+n+1}
  &=\frac{\phi_{m+n+1}}{\phi_{m+n}}\\
  &=\frac{\phi_{m+1}\phi_n-\gamma_2\phi_{m}\phi_{n-1}-\gamma_3\phi_{m-1}\phi_{n-1}-\gamma_3\phi_{m}\phi_{n-2}}{\phi_m\phi_n-\gamma_2\phi_{m-1}\phi_{n-1}-\gamma_3\phi_{m-2}\phi_{n-1}-\gamma_3\phi_{m-1}\phi_{n-2}}\\
  &=\frac{(\xi_{m+1}\xi_m\xi_n-\gamma_2\xi_m-\gamma_3+\xi_m(\xi_{n+1}\xi_n+\gamma_1\xi_n+\gamma_2))\phi_{m-1}\phi_{n-1}}{(\xi_m\xi_n-\gamma_2+(\xi_{m+1}\xi_m+\gamma_1\xi_m+\gamma_2)+(\xi_{n+1}\xi_n+\gamma_1\xi_n+\gamma_2))\phi_{m-1}\phi_{n-1}}\\
  &=\frac{\xi_{m+1}\xi_m\xi_n+\xi_{n+1}\xi_m\xi_n+\gamma_1\xi_m\xi_n-\gamma_3}{\xi_{m+1}\xi_m+\xi_{n+1}\xi_n+\xi_m\xi_n+\gamma_1(\xi_m+\xi_n)+\gamma_2}.
\end{align*}
Especially, if $m=n$, then the relation yields
\begin{equation}\label{eq:xi-cubic-rec-1}
  \xi_{2n+1}
  =\frac{2\xi_{n+1}\xi_n^2+\gamma_1\xi_n^2-\gamma_3}{(2\xi_{n+1}+\xi_n)\xi_n+2\gamma_1\xi_n+\gamma_2}
  =\xi_n-\frac{\xi_n^3+\gamma_1\xi_n^2+\gamma_2\xi_n+\gamma_3}{3\xi_n^2+2\gamma_1\xi_n+\gamma_2+2(\xi_{n+1}-\xi_n)\xi_n}.
\end{equation}

Since the right hand side of the relation~\eqref{eq:xi-cubic-rec-1}
depends on $\xi_n$ and $\xi_{n+1}$,
we cannot compute the evolution of $\xi_n$ using only the relation~\eqref{eq:xi-cubic-rec-1}.
Therefore, we seek to find an additional relation:
\begin{align*}
  \xi_{m+n+2}
  &=\frac{\phi_{m+n+2}}{\phi_{m+n+1}}\\
  &=\frac{\phi_{m+1}\phi_{n+1}-\gamma_2\phi_{m}\phi_{n}-\gamma_3\phi_{m-1}\phi_{n}-\gamma_3\phi_{m}\phi_{n-1}}{\phi_{m+1}\phi_n-\gamma_2\phi_{m}\phi_{n-1}-\gamma_3\phi_{m-1}\phi_{n-1}-\gamma_3\phi_{m}\phi_{n-2}}\\
  &=\frac{(\xi_{m+1}\xi_{n+1}\xi_m\xi_n-\gamma_2\xi_m\xi_n-\gamma_3\xi_n-\gamma_3\xi_m)\phi_{m-1}\phi_{n-1}}{(\xi_{m+1}\xi_m\xi_n-\gamma_2\xi_m-\gamma_3+\xi_m(\xi_{n+1}\xi_n+\gamma_1\xi_n+\gamma_2))\phi_{m-1}\phi_{n-1}}\\
  &=\frac{\xi_{m+1}\xi_{n+1}\xi_m\xi_n-\gamma_2\xi_m\xi_n-\gamma_3(\xi_m+\xi_n)}{\xi_{m+1}\xi_m\xi_n+\xi_{n+1}\xi_m\xi_n+\gamma_1\xi_m\xi_n-\gamma_3}.
\end{align*}
Setting $m=n$ yields
\begin{equation*}
  \xi_{2n+2}=\frac{\xi_{n+1}^2\xi_n^2-\gamma_2\xi_n^2-2\gamma_3\xi_n}{2\xi_{n+1}\xi_n^2+\gamma_1\xi_n^2-\gamma_3}.
\end{equation*}
Furthermore, by using the relation \eqref{eq:xi-cubic-rec-1},
we can rewrite this relation as
\begin{align*}
  \xi_{2n+2}\xi_{2n+1}
  &=\frac{\xi_{n+1}^2\xi_n-\gamma_2\xi_n-2\gamma_3}{(2\xi_{n+1}+\xi_n)\xi_n+2\gamma_1\xi_n+\gamma_2}\xi_n\\
  &=\left(\frac{\xi_{n+1}^2\xi_n+2\xi_n^3+2\gamma_1\xi_n^2+\gamma_2\xi_n}{(2\xi_{n+1}+\xi_n)\xi_n+2\gamma_1\xi_n+\gamma_2}+2(\xi_{2n+1}-\xi_n)\right)\xi_n\\
  &=\left(\frac{\xi_{n+1}^2\xi_n-2\xi_{n+1}\xi_n^2+\xi_n^3}{(2\xi_{n+1}+\xi_n)\xi_n+2\gamma_1\xi_n+\gamma_2}+2\xi_{2n+1}-\xi_n\right)\xi_n\\
  &=\frac{((\xi_{n+1}-\xi_n)\xi_n)^2}{3\xi_n^2+2\gamma_1\xi_n+\gamma_2+2(\xi_{n+1}-\xi_n)\xi_n}+2\xi_{2n+1}\xi_n-\xi_n^2
\end{align*}
Hence, we obtain
\begin{equation*}
  (\xi_{2n+2}-\xi_{2n+1})\xi_{2n+1}
  =\frac{((\xi_{n+1}-\xi_n)\xi_n)^2}{3\xi_n^2+2\gamma_1\xi_n+\gamma_2+2(\xi_{n+1}-\xi_n)\xi_n}-(\xi_{2n+1}-\xi_n)^2.
\end{equation*}
Let us introduce a new variable $\rho_{n}\coloneq (\xi_{n+1}-\xi_n)\xi_n$.
Then, we obtain the following system of difference equations
\begin{subequations}\label{eq:xi-rho}
\begin{align}
  \xi_{2n+1}
  &=\xi_n-\frac{\xi_n^3+\gamma_1\xi_n^2+\gamma_2\xi_n+\gamma_3}{3\xi_n^2+2\gamma_1\xi_n+\gamma_2+2\rho_n},\\
  \rho_{2n+1}
  &=\frac{\rho_n^2}{3\xi_n^2+2\gamma_1\xi_n+\gamma_2+2\rho_n}-(\xi_{2n+1}-\xi_n)^2.
\end{align}
\end{subequations}
Note that the initial values are fixed to $\xi_1=\frac{\phi_1}{\phi_0}=-\gamma_1$ and
$\rho_1=\frac{\phi_2}{\phi_0}-\left(\frac{\phi_1}{\phi_0}\right)^2=-\gamma_2$.

\subsection{The solution to the initial value problem}
To obtain the solution for any initial value,
consider a variable transformation $\xi_n=\eta_n-\eta_1+\xi_1=\eta_n-\eta_1-\gamma_1$,
where one can choose any constant $\eta_1\in\mathbf C$.
Then, we can verify that the relations
\begin{subequations}
\begin{gather}
  \xi_n^3+\gamma_1\xi_n^2+\gamma_2\xi_n+\gamma_3
  =\eta_n^3+c_1\eta_n^2+c_2\eta_n+c_3,\label{eq:cubic-rec-1}\\
  3\xi_n^2+2\gamma_1\xi_n+\gamma_2
  =3\eta_n^2+2c_1\eta_n+c_2\label{eq:cubic-rec-2}
\end{gather}
\end{subequations}
hold, where
\begin{subequations}\label{eq:cubic-coef}
  \begin{gather}
    c_1\coloneq -3\eta_1-2\gamma_1,\\
    c_2\coloneq 3\eta_1^2+4\gamma_1\eta_1+\gamma_1^2+\gamma_2,\\
    c_3\coloneq -\eta_1^3-2\gamma_1\eta_1^2-(\gamma_1^2+\gamma_2)\eta_1-\gamma_1\gamma_2+\gamma_3.
  \end{gather}
\end{subequations}
Note that the relation~\eqref{eq:cubic-rec-2} is derived immediately
by differentiating both sides of the relation~\eqref{eq:cubic-rec-1}
with respect to $\xi_n$ or $\eta_n$.
Hence, the system~\eqref{eq:xi-rho} is transformed into
\begin{align*}
  \eta_{2n+1}
  &=\eta_n-\frac{\eta_n^3+c_1\eta_n^2+c_2\eta_n+c_3}{3\eta_n^2+2c_1\eta_n+c_2+2\rho_n},\\
  \rho_{2n+1}
  &=\frac{\rho_n^2}{3\eta_n^2+2c_1\eta_n+c_2+2\rho_n}-(\eta_{2n+1}-\eta_n)^2.
\end{align*}
Let $z_n\coloneq\eta_{2^{n+1}-1}$ and $r_n\coloneq \rho_{2^{n+1}-1}$, then the system becomes
\begin{subequations}\label{eq:zr-rec}
  \begin{align}
    z_{n+1}
    &=z_n-\frac{z_n^3+c_1z_n^2+c_2z_n+c_3}{3z_n^2+2c_1z_n+c_2+2r_n},\label{eq:zr-rec-1}\\
    r_{n+1}
    &=\frac{r_n^2}{3z_n^2+2c_1z_n+c_2+2r_n}-(z_{n+1}-z_n)^2.
  \end{align}
\end{subequations}
The equation~\eqref{eq:zr-rec-1} is very similar to the equation~\eqref{eq:newton-cubic},
which is the Newton-Raphson iteration for cubic equations.
The term ${}+2r_n$ in the denominator of the equation~\eqref{eq:zr-rec-1}
is considered to be a correction term for making the equation solvable.
One can choose any initial value $z_0=\eta_1$ and $r_0$ is determined as
$r_0=\rho_1=-\gamma_2$, which depends on $z_0$, $c_1$ and $c_2$ (see below).
The solution to the initial value problem is given by
\begin{subequations}\label{eq:zr-presol}
\begin{gather}
  z_n=\eta_{2^{n+1}-1}
  =\xi_{2^{n+1}-1}+\eta_1+\gamma_1=\frac{\phi_{2^{n+1}-1}}{\phi_{2^{n+1}-2}}+z_0+\gamma_1,\\
  r_n=\rho_{2^{n+1}-1}
  =(\xi_{2^{n+1}}-\xi_{2^{n+1}-1})\xi_{2^{n+1}-1}
  =\frac{\phi_{2^{n+1}}}{\phi_{2^{n+1}-2}}-\left(\frac{\phi_{2^{n+1}-1}}{\phi_{2^{n+1}-2}}\right)^2.
\end{gather}
\end{subequations}

Let us solve the four-term recurrence relation~\eqref{eq:four-term}
in terms of roots of the cubic equation $z^3+c_1z^2+c_2z+c_3=0$.
First, we calculate the reciprocal relations of \eqref{eq:cubic-coef}:
\begin{subequations}\label{eq:gamma-c}
  \begin{gather}
    \gamma_1=-\frac{3z_0+c_1}{2},\\
    \gamma_2=\frac{3z_0^2+2c_1z_0-c_1^2+4c_2}{4},\\
    \gamma_3=-\frac{z_0^3+c_1z_0^2-(c_1^2-4c_2)z_0-c_1^3+4c_1c_2-8c_3}{8}.
  \end{gather}
\end{subequations}
Let $\lambda_0$, $\lambda_1$ and $\lambda_2$ be the three roots of the cubic equation.
Substituting $c_1=-\lambda_0-\lambda_1-\lambda_2$,
$c_2=\lambda_0\lambda_1+\lambda_1\lambda_2+\lambda_2\lambda_0$,
and $c_3=-\lambda_0\lambda_1\lambda_2$ into \eqref{eq:gamma-c} yields
\begin{gather*}
  \gamma_1=-\frac{3z_0-\lambda_0-\lambda_1-\lambda_2}{2},\\
  \gamma_2=\frac{3z_0^2+(-2\lambda_0-2\lambda_1-2\lambda_2)z_0-\lambda_0^2-\lambda_1^2-\lambda_2^2+2\lambda_0\lambda_1+2\lambda_1\lambda_2+2\lambda_2\lambda_0}{4},\\
  \begin{multlined}[t]
    \gamma_3=-\frac{1}{8}(z_0^3+(-\lambda_0-\lambda_1-\lambda_2)z_0^2+(-\lambda_0^2-\lambda_1^2-\lambda_2^2+2\lambda_0\lambda_1+2\lambda_1\lambda_2+2\lambda_2\lambda_0)z_0\\
    +\lambda_0^3+\lambda_1^3+\lambda_2^3-\lambda_0^2\lambda_1-\lambda_0\lambda_1^2-\lambda_1^2\lambda_2-\lambda_1\lambda_2^2-\lambda_2^2\lambda_0-\lambda_2\lambda_0^2+2\lambda_0\lambda_1\lambda_2).
  \end{multlined}
\end{gather*}
Then, one can find that the characteristic equation of the four-term recurrence
relation~\eqref{eq:four-term} $z^3+\gamma_1z^2+\gamma_2z+\gamma_3=0$ is factored into
$(z-\Lambda_0)(z-\Lambda_1)(z-\Lambda_2)=0$, where
\begin{equation*}
  \Lambda_0\coloneq\frac{z_0+\lambda_0-\lambda_1-\lambda_2}{2},\quad
  \Lambda_1\coloneq\frac{z_0+\lambda_1-\lambda_2-\lambda_0}{2},\quad
  \Lambda_2\coloneq\frac{z_0+\lambda_2-\lambda_0-\lambda_1}{2}.
\end{equation*}

We divide into cases.

\paragraph{Case 1}
If all the roots $\lambda_0$, $\lambda_1$, and $\lambda_2$ are distinct,
then $\phi_n$ is of the form
\begin{equation*}
  \phi_n=C_0\Lambda_0^n+C_1\Lambda_1^n+C_2\Lambda_2^n,
\end{equation*}
where the coefficients $C_0$, $C_1$ and $C_2$ are determined by
$\phi_0=1$, $\phi_1=-\gamma_1=\Lambda_0+\Lambda_1+\Lambda_2$
and $\phi_2=\gamma_1^2-\gamma_2=\Lambda_0^2+\Lambda_1^2+\Lambda_2^2+\Lambda_0\Lambda_1+\Lambda_1\Lambda_2+\Lambda_2\Lambda_0$, i.e.
\begin{equation*}
  \begin{pmatrix}
    1 & 1 & 1 \\
    \Lambda_0 & \Lambda_1 & \Lambda_2\\
    \Lambda_0^2 & \Lambda_1^2 & \Lambda_2^2
  \end{pmatrix}
  \begin{pmatrix}
    C_0\\
    C_1\\
    C_2
  \end{pmatrix}
  =
  \begin{pmatrix}
    1 \\
    \Lambda_0+\Lambda_1+\Lambda_2\\
    \Lambda_0^2+\Lambda_1^2+\Lambda_2^2+\Lambda_0\Lambda_1+\Lambda_1\Lambda_2+\Lambda_2\Lambda_0
  \end{pmatrix}.
\end{equation*}
The solution of this equation is
\begin{equation*}
  C_0=\frac{\Lambda_0^2}{(\Lambda_0-\Lambda_1)(\Lambda_0-\Lambda_2)},\quad
  C_1=\frac{\Lambda_1^2}{(\Lambda_1-\Lambda_0)(\Lambda_1-\Lambda_2)},\quad
  C_2=\frac{\Lambda_2^2}{(\Lambda_2-\Lambda_0)(\Lambda_2-\Lambda_1)}.
\end{equation*}
In this case, the solution~\eqref{eq:zr-presol} becomes
\begin{align*}
  z_n
  &=\frac{C_0\Lambda_0^{2^{n+1}-1}+C_1\Lambda_1^{2^{n+1}-1}+C_2\Lambda_2^{2^{n+1}-1}}{C_0\Lambda_0^{2^{n+1}-2}+C_1\Lambda_1^{2^{n+1}-2}+C_2\Lambda_2^{2^{n+1}-2}}-\frac{z_0-\lambda_0-\lambda_1-\lambda_2}{2}\\
  &=\frac{C_0\lambda_0\Lambda_0^{2^{n+1}-2}+C_1\lambda_1\Lambda_1^{2^{n+1}-2}+C_2\lambda_2\Lambda_2^{2^{n+1}-2}}{C_0\Lambda_0^{2^{n+1}-2}+C_1\Lambda_1^{2^{n+1}-2}+C_2\Lambda_2^{2^{n+1}-2}},\\
  r_n&=
    \frac{C_0\Lambda_0^{2^{n+1}}+C_1\Lambda_1^{2^{n+1}}+C_2\Lambda_2^{2^{n+1}}}{C_0\Lambda_0^{2^{n+1}-2}+C_1\Lambda_1^{2^{n+1}-2}+C_2\Lambda_2^{2^{n+1}-2}}
    -\left(\frac{C_0\Lambda_0^{2^{n+1}-1}+C_1\Lambda_1^{2^{n+1}-1}+C_2\Lambda_2^{2^{n+1}-1}}{C_0\Lambda_0^{2^{n+1}-2}+C_1\Lambda_1^{2^{n+1}-2}+C_2\Lambda_2^{2^{n+1}-2}}\right)^2.
\end{align*}
If $|\Lambda_0|>|\Lambda_1|\ge |\Lambda_2|$, then
$z_n\to \lambda_0$ and $r_n\to\Lambda_0^2-\Lambda_0^2=0$
as $n\to\infty$
and the convergence rate is quadratic.
If $|\Lambda_0|=|\Lambda_1|\ge |\Lambda_2|$ and $z_0$ is not a root of the cubic equation
$\lambda_0$, $\lambda_1$ or $\lambda_2$, then $z_n$ and $r_n$ does not converge.

In addition, we can prove the following theorem.

\begin{theorem}\label{thm:quad-where-converge}
  Suppose that the roots $\lambda_0$, $\lambda_1$, $\lambda_2$ are real
  and $\lambda_0<\lambda_1<\lambda_2$.
  Then, as $n\to\infty$,
  \begin{itemize}
  \item if $z_0<\lambda_1$ then $z_n\to \lambda_0$;
  \item if $z_0=\lambda_1$ then $z_n=\lambda_1$ for all $n$;
  \item if $z_0>\lambda_1$ then $z_n\to \lambda_2$.
  \end{itemize}
\end{theorem}
\begin{proof}
  Since $z_n=\lambda_1$ is a fixed point of the system \eqref{eq:zr-rec-1},
  the second assertion obviously holds.
  The first and third assertions are proved as follows.
  Since $\Lambda_2-\Lambda_1=\lambda_2-\lambda_1$ and $\Lambda_1-\Lambda_0=\lambda_1-\lambda_0$
  hold, and from the assumption $\lambda_0<\lambda_1<\lambda_2$,
  we obtain the inequality $\Lambda_0<\Lambda_1<\Lambda_2$.
  Therefore, $\max\{|\Lambda_0|, |\Lambda_1|, |\Lambda_2|\}$ is $|\Lambda_0|$ or $|\Lambda_2|$,
  not $|\Lambda_1|$.
  Let $\epsilon=\frac{\lambda_1-z_0}{2}$,
  then
  \begin{align*}
    \Lambda_0&=\frac{z_0+\lambda_0-\lambda_1-\lambda_2}{2}=-\frac{\lambda_2-\lambda_0+2\epsilon}{2}=-\frac{\lambda_2-\lambda_0}{2}-\epsilon,\\
   \Lambda_2&=\frac{z_0+\lambda_2-\lambda_0-\lambda_1}{2}=\frac{\lambda_2-\lambda_0-2\epsilon}{2}=\frac{\lambda_2-\lambda_0}{2}-\epsilon.
  \end{align*}
  Hence, if $z_0<\lambda_1$ i.e. $\epsilon>0$ then $|\Lambda_0|>|\Lambda_2|$,
  and if $z_0>\lambda_1$ i.e. $\epsilon<0$ then $|\Lambda_0|<|\Lambda_2|$.
\end{proof}

\paragraph{Case 2}
If $\lambda_0\ne \lambda_1$ and $\lambda_1=\lambda_2$, then,
since $\Lambda_1=\Lambda_2=\frac{z_0-\lambda_0}{2}$,
$\phi_n$ is of the form
\begin{equation*}
  \phi_n=C_0\Lambda_0^n+C_1\Lambda_1^n+C_2 n\Lambda_1^{n-1},
\end{equation*}
where the coefficients $C_0$, $C_1$ and $C_2$ are determined by
$\phi_0=1$, $\phi_1=-\gamma_1=\Lambda_0+2\Lambda_1$ and $\phi_2=\gamma_1^2-\gamma_2=\Lambda_0^2+3\Lambda_1^2+2\Lambda_0\Lambda_1$, i.e.
\begin{equation*}
  \begin{pmatrix}
    1 & 1 & 0 \\
    \Lambda_0 & \Lambda_1 & 1\\
    \Lambda_0^2 & \Lambda_1^2 & 2\Lambda_1
  \end{pmatrix}
  \begin{pmatrix}
    C_0\\
    C_1\\
    C_2
  \end{pmatrix}
  =
  \begin{pmatrix}
    1 \\
    \Lambda_0+2\Lambda_1\\
    \Lambda_0^2+3\Lambda_1^2+2\Lambda_0\Lambda_1
  \end{pmatrix}.
\end{equation*}
The solution of this equation is
\begin{equation*}
  C_0=\frac{\Lambda_0^2}{(\Lambda_0-\Lambda_1)^2},\quad
  C_1=1-\frac{\Lambda_0^2}{(\Lambda_1-\Lambda_0)^2},\quad
  C_2=\frac{\Lambda_1^2}{\Lambda_1-\Lambda_0}.
\end{equation*}
In this case, the solution~\eqref{eq:zr-presol} becomes
\begin{align*}
  z_n
  &=\frac{C_0\Lambda_0^{2^{n+1}-1}+(C_1\Lambda_1+(2^{n+1}-1)C_2)\Lambda_1^{2^{n+1}-2}}{C_0\Lambda_0^{2^{n+1}-2}+(C_1\Lambda_1+(2^{n+1}-2)C_2)\Lambda_1^{2^{n+1}-3}}-\frac{z_0-\lambda_0-2\lambda_1}{2}\\
  &=\frac{C_0\lambda_0\Lambda_0^{2^{n+1}-2}+(C_1\lambda_1\Lambda_1+C_2\Lambda_1+(2^{n+1}-2)C_2\lambda_1)\Lambda_1^{2^{n+1}-3}}{C_0\Lambda_0^{2^{n+1}-2}+(C_1\Lambda_1+(2^{n+1}-2)C_2)\Lambda_1^{2^{n+1}-3}},\\
  r_n&=
  \begin{multlined}[t]
    \frac{C_0\Lambda_0^{2^{n+1}}+(C_1\Lambda_1+2^{n+1}C_2)\Lambda_1^{2^{n+1}-1}}{C_0\Lambda_0^{2^{n+1}-2}+(C_1\Lambda_1+(2^{n+1}-2)C_2)\Lambda_1^{2^{n+1}-3}}\\
    -\left(\frac{C_0\Lambda_0^{2^{n+1}-1}+(C_1\Lambda_1+(2^{n+1}-1)C_2)\Lambda_1^{2^{n+1}-2}}{C_0\Lambda_0^{2^{n+1}-2}+(C_1\Lambda_1+(2^{n+1}-2)C_2)\Lambda_1^{2^{n+1}-3}}\right)^2.
  \end{multlined}
\end{align*}
If $|\Lambda_0|>|\Lambda_1|$, then
$z_n\to \lambda_0$ and $r_n\to 0$ as $n\to\infty$
and the convergence rate is quadratic.
If $|\Lambda_0|<|\Lambda_1|$,
then, since we can verify that the relation
\begin{equation*}
  \frac{C_1\lambda_1\Lambda_1+C_2\Lambda_1+(2^{n+1}-2)C_2\lambda_1}{C_1\Lambda_1+(2^{n+1}-2)C_2}
  =\frac{\lambda_0 z_0+2^{n+2}(\lambda_0-\lambda_1)\lambda_1-\lambda_0^2}{z_0+2^{n+2}(\lambda_0-\lambda_1)-\lambda_0}
\end{equation*}
holds, we can conclude that $z_n\to \lambda_1$ and $r_n\to 0$ as
$n\to\infty$
and the convergence rate is linear,
where $\lim_{n\to\infty}\frac{|z_{n+1}-\lambda_1|}{|z_n-\lambda_1|}=\frac{1}{2}$ and
$\lim_{n\to\infty}\frac{|r_{n+1}|}{|r_n|}=\frac{1}{4}$.
If $|\Lambda_0|=|\Lambda_1|$, $z_0 \ne \lambda_0$ and $z_0\ne \lambda_1$,
then $z_n$ and $r_n$ does not converge.

The following theorem is proved in the same way as
the proof of Theorem~\ref{thm:quad-where-converge}.

\begin{theorem}\label{thm:multiple-where-conberge}
  Suppose that $\lambda_0$, $\lambda_1$, $\lambda_2$ are real,
  $\lambda_0\ne \lambda_1$, and $\lambda_1=\lambda_2$.
  Then, as $n\to \infty$,
  \begin{itemize}
  \item if $\lambda_0<\lambda_1$ and $z_0<\lambda_1$ then $z_n\to \lambda_0$;
  \item if $\lambda_0<\lambda_1$ and $z_0\ge\lambda_1$ then $z_n\to \lambda_1$;
  \item if $\lambda_0>\lambda_1$ and $z_0>\lambda_1$ then $z_n\to \lambda_0$;
  \item if $\lambda_0>\lambda_1$ and $z_0\le\lambda_1$ then $z_n\to \lambda_1$.
  \end{itemize}
\end{theorem}

\paragraph{Case 3}
If $\lambda_0=\lambda_1=\lambda_2\eqcolon\lambda$, then,
since $\Lambda_0=\Lambda_1=\Lambda_2=\frac{z_0-\lambda}{2}\eqcolon\Lambda$,
$\phi_n$ is of the form
\begin{equation*}
  \phi_n=C_0\Lambda^n+C_1n\Lambda^{n-1}+C_2 \frac{n(n-1)}{2}\Lambda^{n-2},
\end{equation*}
where the coefficients $C_0$, $C_1$ and $C_2$ are determined by
$\phi_0=1$, $\phi_1=-\gamma_1=3\Lambda$ and $\phi_2=\gamma_1^2-\gamma_2=6\Lambda^2$, i.e.
\begin{equation*}
  \begin{pmatrix}
    1 & 0 & 0 \\
    \Lambda & 1 & 0\\
    \Lambda^2 & 2\Lambda & 1
  \end{pmatrix}
  \begin{pmatrix}
    C_0\\
    C_1\\
    C_2
  \end{pmatrix}
  =
  \begin{pmatrix}
    1 \\
    3\Lambda\\
    6\Lambda^2
  \end{pmatrix}.
\end{equation*}
The solution of this equation is
\begin{equation*}
  C_0=1,\quad
  C_1=2\Lambda,\quad
  C_2=\Lambda^2.
\end{equation*}
In this case, the solution~\eqref{eq:zr-presol} becomes
\begin{align*}
  z_n
  &=\frac{(2^{n+2}-1+(2^{n+1}-1)(2^{n}-1))\Lambda}{2^{n+2}-3+(2^{n+1}-3)(2^{n}-1)}-\frac{z_0-3\lambda}{2}\\
  &=\frac{2^{n+2}\lambda+2\Lambda-3\lambda+(2^{n+1}\lambda+2\Lambda-3\lambda)(2^{n}-1)}{2^{n+2}-3+(2^{n+1}-3)(2^{n}-1)},\\
  r_n
  &=\frac{(2^{n+2}+1+(2^{n+1}-1)2^{n})\Lambda^2}{2^{n+2}-3+(2^{n+1}-3)(2^{n}-1)}
  -\left(\frac{(2^{n+2}-1+(2^{n+1}-1)(2^{n}-1))\Lambda}{2^{n+2}-3+(2^{n+1}-3)(2^{n}-1)}\right)^2.
\end{align*}
Hence, $z_n\to \lambda$ and $r_n\to 0$ as $n\to\infty$
and the convergence rate is linear,
where $\lim_{n\to\infty}\frac{|z_{n+1}-\lambda|}{|z_n-\lambda|}=\frac{1}{2}$ and
$\lim_{n\to\infty}\frac{|r_{n+1}|}{|r_n|}=\frac{1}{4}$.

\subsection{Numerical examples}

The following cases correspond to the cases in the previous subsection.
The values were computed in double precision using Python 3.
Note again that the value of $r_0$ depends on $z_0$, $c_1$ and $c_2$: $r_0=-\gamma_2=-\frac{3z_0^2+2c_1z_0-c_1^2+4c_2}{4}$.

\begin{table}[t]
  \caption{Computation results for the cubic equation $(z+4)(z-1)(z-2)=z^3+z^2-10z+8=0$.
  The initial values are $z_0=-10$ (left) and $z_0=1.1$ (right).}
  \label{tab:case1}
  \centering

  \ttfamily
  \footnotesize
  \medskip
  \begin{minipage}{.49\textwidth}
    \begin{tabular}{rrr}
      \hline
      $n$ & $z_n$ & $r_n$\\\hline
      0 & -10.000000000000000 & -59.750000000000000\\
      1 & -4.737541528239203 & -3.972123652056819\\
      2 & -4.023220863117919 & -0.114955702599983\\
      3 & -4.000020394369582 & -0.000101878191844\\
      4 & -4.000000000013991 & -0.000000000069960\\
      5 & -4.000000000000000 & -0.000000000000000\\
      6 & -4.000000000000000 & 0.000000000000000\\
      7 & -4.000000000000000 & 0.000000000000000\\
      8 & -4.000000000000000 & 0.000000000000000\\
      9 & -4.000000000000000 & 0.000000000000000\\
      10 & -4.000000000000000 & 0.000000000000000\\\hline
    \end{tabular}
  \end{minipage}\hfill
  \begin{minipage}{.49\textwidth}
    \begin{tabular}{rrr}
      \hline
      $n$ & $z_n$ & $r_n$\\\hline
      0 & 1.100000000000000 & 8.792500000000000\\
      1 & 1.134215430488259 & 5.761636330033983\\
      2 & 1.212192345618043 & 4.332736842090070\\
      3 & 1.370666542852025 & 3.389273687697018\\
      4 & 1.613648691839280 & 2.168857411309579\\
      5 & 1.861185520779511 & 0.813617415734770\\
      6 & 1.983219552286651 & 0.100401102854631\\
      7 & 1.999764038372010 & 0.001415714090050\\
      8 & 1.999999953598109 & 0.000000278411343\\
      9 & 1.999999999999998 & 0.000000000000011\\
      10 & 2.000000000000000 & 0.000000000000000\\\hline
    \end{tabular}
  \end{minipage}
\end{table}
\paragraph{Case 1}
Consider the cubic equation $(z+4)(z-1)(z-2)=z^3+z^2-10z+8=0$.
Table~\ref{tab:case1} is computation results of the system~\eqref{eq:zr-rec}
for this equation. The convergence rate is indeed quadratic.
For $z_0=1.1$, though $|z_0-1|$ is smaller than $|z_0+4|$ and $|z_0-2|$,
$z_n$ goes to 2 (see Theorem~\ref{thm:quad-where-converge}).

\begin{table}[t]
  \caption{Computation results for the cubic equation $(z+4)(z-1)^2=z^3+2z^2-7z+4=0$.
  The initial values are $z_0=0.9$ (left) and $z_0=1.1$ (right).}
  \label{tab:case2}
  \centering

  \ttfamily
  \footnotesize
  \medskip
  \begin{minipage}{.49\textwidth}
     \begin{tabular}{rrr}
      \hline
      $n$ & $z_n$ & $r_n$\\\hline
      0 & 0.900000000000000 & 6.492500000000000\\
      1 & 0.895921764461090 & 3.508310979314073\\
      2 & 0.887095029874565 & 2.048450171247236\\
      3 & 0.866370988166584 & 1.395451403926513\\
      4 & 0.808753710259613 & 1.287826511304223\\
      5 & 0.581198636424688 & 2.093986246044651\\
      6 & -0.945986990298644 & 6.001525302845456\\
      7 & -3.908537536032578 & 0.449822261884304\\
      8 & -3.999993354829612 & 0.000033257605959\\
      9 & -3.999999999999983 & 0.000000000000085\\
      10 & -4.000000000000000 & 0.000000000000000\\\hline
    \end{tabular}
  \end{minipage}\hfill
  \begin{minipage}{.49\textwidth}
    \begin{tabular}{rrr}
      \hline
      $n$ & $z_n$ & $r_n$\\\hline
      0 & 1.100000000000000 & 5.992500000000000\\
      1 & 1.096081444487130 & 2.759113054450161\\
      2 & 1.088851212112513 & 1.169920725404437\\
      3 & 1.076497685250043 & 0.420726496487894\\
      4 & 1.058204951031817 & 0.108663135551645\\
      5 & 1.037037021593743 & 0.014137594891871\\
      6 & 1.019881645542941 & 0.000201947113096\\
      7 & 1.009980464499531 & -0.000097829886169\\
      8 & 1.004985434704009 & -0.000024854527482\\
      9 & 1.002491475246450 & -0.000006207448904\\
      10 & 1.001245427328125 & -0.000001551089229\\\hline
    \end{tabular}
  \end{minipage}
\end{table}
\paragraph{Case 2}
Consider the cubic equation $(z+4)(z-1)^2=z^3+2z^2-7z+4=0$.
Table~\ref{tab:case2} is computation results of the system~\eqref{eq:zr-rec}
for this equation.
We can confirm that the destination of $z_n$ is
as stated in Theorem~\ref{thm:multiple-where-conberge}.
The convergence rate is quadratic for $z_0<1$ and linear for $z_0>1$,
where $\frac{|z_{n+1}-1|}{|z_n-1|}\approx\frac{1}{2}$ and $\frac{|r_{n+1}|}{|r_n|}\approx\frac{1}{4}$
for large $n$.

\begin{table}[t]
  \caption{A computation result for the cubic equation $(z-1)^3=z^3-3z^2+3z-1=0$.
  The initial value is $z_0=1.1$.}
  \label{tab:case3}
  \centering

  \ttfamily
  \footnotesize
  \medskip
  \begin{tabular}{rrr}
    \hline
    $n$ & $z_n$ & $r_n$\\\hline
    0 & 1.100000000000000 & -0.007500000000000\\
    1 & 1.033333333333339 & -0.000694444444443\\
    2 & 1.014285714285861 & -0.000114795918363\\
    3 & 1.006666666667307 & -0.000023611111108\\
    4 & 1.003225806447293 & -0.000005365504721\\
    5 & 1.001587301576629 & -0.000001279446954\\
    6 & 1.000787401570636 & -0.000000312422489\\
    7 & 1.000392156874934 & -0.000000077193860\\
    8 & 1.000195693252556 & -0.000000019186196\\
    9 & 1.000097745995403 & -0.000000004782721\\
    10 & 1.000048819760003 & -0.000000001196000\\\hline
  \end{tabular}
\end{table}
\paragraph{Case 3}
Consider the cubic equation $(z-1)^3=z^3-3z^2+3z-1=0$.
Table~\ref{tab:case3} is a computation result of the system~\eqref{eq:zr-rec}
for this equation.
The convergence rate is indeed linear,
where $\frac{|z_{n+1}-1|}{|z_n-1|}\approx\frac{1}{2}$ and $\frac{|r_{n+1}|}{|r_n|}\approx\frac{1}{4}$ for large $n$.

\section{Concluding remarks}\label{sec:concluding-remarks}
In this paper, we extended the method to obtain a closed form solution to the initial value problem
of the Newton--Raphson iteration for quadratic equations to the case of cubic equations
and derived a solvable iterative method similar to the Newton--Raphson iteration for cubic equations.
We showed that the solution to this novel system also converges to a root of a given cubic equation
and the convergence rate is quadratic in many cases and verified this fact by numerical examples.

Originally, Kondo--Nakamura's motivation was the relationship between numerical algorithms and
discrete integrable systems~\cite{nakamura2001aaa}.
The determinants of tridiagonal matrices they considered appear as a special case of the Chebyshev polynomials,
which is one of the classical orthogonal polynomials~\cite{chihara1978iop}.
The author is studying discrete integrable systems from the point of view of orthogonal functions,
and a work on the biorthogonal polynomials~\cite{maeda2017nuh} leads to this study.
We expect that the theory of orthogonal functions may derive more solvable iterative methods.

\section*{Acknowledgements}
This work was supported by JSPS KAKENHI Grant Number JP21K13837.




\end{document}